\documentclass[12pt, draftclsnofoot, onecolumn]{IEEEtran}
\usepackage{amsmath}
\usepackage{amssymb}
\usepackage{graphicx}
\usepackage{epsfig}
\usepackage{subfigure}
\usepackage{mathrsfs}
\usepackage{fancyhdr}

\newtheorem{theorem}{Theorem}

\newcommand{\PP}{\mathbb{P}}
\newcommand{\E}{\mathbb{E}}

\newcommand{\sir}{{\rm SIR}}
\newcommand{\sinr}{{\rm SINR}}

\newcommand{\dd}{{\rm d}}

\newcommand{\TR}{T_{\mathrm{FR}}}

\begin{document}

\title{Analytical Evaluation of Fractional Frequency Reuse for Heterogeneous Cellular Networks}

\author{\IEEEauthorblockN{Thomas D. Novlan, Radha Krishna Ganti, Arunabha Ghosh, Jeffrey G. Andrews\\}
}

\maketitle
\pagestyle{fancy}
\thispagestyle{empty}
\fancyhf{} 
\fancyhead[R]{\thepage}
\renewcommand{\headrulewidth}{0pt}
\renewcommand{\footrulewidth}{0pt}

\begin{abstract}
Interference management techniques are critical to the performance of heterogeneous cellular networks, which will have dense and overlapping coverage areas, and experience high levels of interference. Fractional frequency reuse (FFR) is an attractive interference management technique due to its low complexity and overhead, and significant coverage improvement for low-percentile (cell-edge) users. Instead of relying on system simulations based on deterministic access point locations, this paper instead proposes an analytical model for evaluating Strict FFR and Soft Frequency Reuse (SFR) deployments based on the spatial Poisson point process. Our results both capture the non-uniformity of heterogeneous deployments and produce tractable expressions which can be used for system design with Strict FFR and SFR. We observe that the use of Strict FFR bands reserved for the users of each tier with the lowest average $\sinr$ provides the highest gains in terms of coverage and rate, while the use of SFR allows for more efficient use of shared spectrum between the tiers, while still mitigating much of the interference. Additionally, in the context of multi-tier networks with closed access in some tiers, the proposed framework shows the impact of cross-tier interference on closed access FFR, and informs the selection of key FFR parameters in open access.
\end{abstract}

\let\thefootnote\relax\footnotetext{T. D. Novlan, R. K. Ganti, and J. G. Andrews are with the Wireless Networking and Communications Group, the University of Texas at Austin. A. Ghosh is with AT\&T Laboratories. The contact author is J. G. Andrews. Email: jandrews@ece.utexas.edu. This research has been supported by AT\&T Laboratories. Date revised: June 23, 2011} 

\linespread{2}
\section{\label{sec:intro} Introduction}
Modern cellular network deployments are currently transitioning from largely homogeneous (one-tier) voice-centric deployments to highly heterogeneous data-centric networks comprised of different classes (tiers) of access points \cite{Qualcomm2011}. These include operator-deployed picocells and distributed antenna systems \cite{Pico2007,Saleh1987,Zhang2008}, and home user-deployed femtocells \cite{Chand2008}. 

Performance analysis of these networks is much more involved than for a single-tier network because of the need to account for inter-cell and cross-tier interference and the non-uniformity of the access point deployments arising from both topographic and economic reasons. A further complication in heterogeneous network analysis arises from different user association policies. As a result, there is a need for new and general models for analyzing the important metrics of coverage and rate in the context of these multi-tier networks. While prior work has relied on simulations based on deterministic models of AP locations, these have not led to general or tractable solutions. In this paper, instead, we model the AP locations as a Poisson point process (PPP) \cite{Baccelli1995,Brown2000,Haenggi2009}. This modeling approach has been recently applied to the analysis of cellular networks due to the ability to derive tractable expressions for coverage and rate both for one-tier \cite{Andrews2010} and very recently, heterogeneous networks \cite{Dhillon2011,ITA2011,MukherjeeLTE2011,MukherjeeWCDMA2011}. 

\subsection{Fractional Frequency Reuse}
Faced with increased traffic demands in interference-limited cellular networks, fractional frequency reuse (FFR) is an attractive strategy due to its low complexity of implementation and its significant gains for the bottom percentile of mobile users. Recently, FFR has been included in fourth generation (4G) wireless standards including WiMAX 2 (802.16m) and 3GPP-LTE since release 8 \cite{Intel2010}. This work extends our novel analytical model of FFR in the downlink of a cellular network with a single-tier of base stations using the PPP model developed in \cite{Novlan2011,NovlanGCOM2011} to a general multi-tier network with closed and open access between the tiers. This allows the development of tractable expressions for the $\sinr$ distributions to be derived as a function of the FFR parameters which can be utilized for the system design of these networks. 

We will consider the two most common types of FFR: \textit{Strict FFR} and \textit{Soft Frequency Reuse} (SFR). Under Strict FFR, which extends the traditional frequency reuse used extensively in current cellular networks \cite{Begain2002,Sternad2003}, users in the interior of a cell are allocated a common sub-band of frequencies $f_c$ while at the cell-edge, users are allocated separate subbands partitioned across cells with a reuse factor of $\Delta$.  The left sub-figure in Fig. \ref{fig:FFRmodels}(a) illustrates potential Strict FFR allocations with $\Delta = 3$ in which edge users are given frequency resources corresponding to subbands $f_1, f_2$, or $f_3$. The primary advantage of Strict FFR is the significant reduction in interference for edge users, although there is a loss in spectral efficiency since each cell cannot fully utilize all $\Delta + 1$ subbands \cite{Novlan2010}.

The right sub-figure in Fig. \ref{fig:FFRmodels}  illustrates the frequency and transmit power allocation for SFR. Edge users are allocated bandwidth subbands with a reuse factor of $\Delta$, but the main difference vs. Strict FFR is that each cell utilizes all $\Delta$ subbands since interior users are allowed to share sub-bands with edge users in other cells. Because cell-edge users share the bandwidth with neighboring cells, their downlinks are typically transmitted with higher power levels in order to reduce the impact of the inter-cell interference \cite{Li1999,Huawei2005}. To accomplish this, a transmit power control factor $\beta \geq 1$ is introduced to create two different classes, $P_{\rm int} = P$ and $P_{\rm edge} = \beta P$, where $P_{\rm int}$ is the transmit power of the base station if user $y$ is an interior user and $P_{\rm edge}$ is the transmit power of the base station if user $y$ is a cell-edge user. The increased interference for edge users under SFR is traded off for greater spectral utilization \cite{Doppler2009}. 

\subsection{Related Work}
Early work on frequency partitioning for two-tier networks is found in \cite{Chand2009}. Their proposed strategy maximizes the spectral efficiency for a minimum QoS requirement and the number of users per tier. They assume that the femtocells are given a separate frequency band from the macrocells, such that there is no cross-tier interference. 

The authors in \cite{HeuiLee2010} consider an adaptive FFR strategy for mitigating inter-femtocell interference while keeping spectral efficiency as high as possible. They vary the size of FFR partitions and transmit power based on the amount of estimated interference. However they use a deterministic model for the femtocells inside of a single building and neglect macrocell or femtocell interference outside of the building. Very recent work in \cite{JuLee2011} considers a deterministic model analysis of the spectral efficiency of femtocells as a function of the femtocell's location in a two-tier network with base stations modeled as a hexagonal grid and femtocells uniformly deployed in each cell. They fix the macrocell FFR sub-band allocations and then consider the spectral efficiency of a femtocell as a function of its distance from the cell center. 

Frequency partitioning between macrocells and femtocells is revisited in \cite{AndrewsM2010}. They propose a model where some sub-bands are reserved for only macrocell or femtocell users in addition to a common group of sub-bands, similar in concept to the proposed Strict FFR model. They also alternately consider partitioning in the time domain. They provide a large number of simulation results based on a deterministic model for the AP locations and motivate a dynamic partitioning based on measured interference levels by users in either tier. 

The two primary user association policies for heterogeneous networks are \textit{closed access} and \textit{open access}. Under closed access, mobiles are restricted from connecting with certain tiers of access points based on system performance metrics or economic or legal factors in some cases \cite{Chand2008}. Open access instead allows users to connect to APs of different tiers based on the association policy, which may be measured signal-to-interference-ratio ($\sir$) or traffic load and can be used as an interference management technique \cite{Xia2010}. The authors in \cite{HanShin2011} consider performance tradeoffs for closed and open femtocell networks. Their analysis uses stochastic geometry tools from \cite{Andrews2010} in order to derive $\sinr$ distributions for different deployment scenarios at the cell edge or interior and for varying femtocell densities. However their analysis is constrained to the interior of a single macrocell and does not consider the effect of inter-cell interference or the use of FFR on the $\sinr$ distributions.

\subsection{Contributions}
In this paper we present the following contributions. First, we extend the framework of \cite{Novlan2011,NovlanGCOM2011} to evaluate the $\sinr$ distributions for users in a downlink $K$-tier network utilizing Strict FFR and SFR. We first consider closed access, which limits users to associate with APs in only one tier, with all the other tiers contributing interference. In addition, by considering a special case relevant to interference-limited networks, the analytical expressions for the $\sinr$ distributions reduce to simple expressions which are a function of the key FFR design parameters, allowing for clear, intuitive comparisons between the reuse strategies and insight into system design. Secondly, we propose a new framework for analyzing coverage for the open access downlink under Strict FFR and SFR in which users may associate with APs in more than one tier. Finally, we provide implications of the analysis to system design for closed and open access networks. The models allow for investigation of FFR parameter selection based on the densities, transmit powers, and resource allocation strategies of the tiers. In the next section, we provide a detailed description of the system model and our assumptions.
 
\section{\label{sec:model}System Model}
We consider an OFDMA cellular downlink with K-tiers of access points (APs). The locations of the base stations and femtocells are modeled as independent spatial Poisson point processes (PPP) \cite{Stoyan1996} of density $\lambda_k$ with independence between the tiers. In other words, for a given PPP, the number of points in a bounded area is a Poisson-distributed random variable and those points are uniformly-distributed within the area. A realization of a three-tier network with Poisson distributed APs and Voronoi cell coverage regions based on strongest received power is given in Fig. \ref{fig:multiTier}. 

Without loss of generality, we assume a typical mobile user at the origin and compute the $\sinr$ for this typical mobile. We assume that the mobile user is served by only one tier at a time and by the closest AP of that tier, which is at a distance $r_k$. Since the underlying APs are distributed as PPPs, it follows that $r_k$ is Rayleigh distributed \cite{Stoyan1996}. We assume that all the access points of the $k^{th}$ transmit with an equal power $P_k$. The path loss exponent is given by $\alpha$, and ${\sigma}^2$ is the noise power. We assume that the  small-scale fading between any interfering AP and the typical mobile in consideration, denoted by $G_z$, is i.i.d exponentially distributed with mean $\mu$ (corresponds to Rayleigh fading). The set of interfering APs in the $k^{th}$ tier is $\mathcal{Z}_k$, i.e. access points that use the same sub-band as the mobile user. We denote the distance between the interfering AP and the mobile node in consideration by $R_z$. 

The associated signal-to-interference-plus-noise-ratio ($\sinr$) is given as
\begin{equation}
\label{eq:SINR}
\sinr = \frac{P_kg_k{r_k}^{-\alpha}}{{\sigma}^2 + \displaystyle\sum_{k=1}^{K}P_kI_k},
\end{equation}
where for an interfering set of $k^{th}$ tier APs $\mathcal{Z}_k$,
\begin{equation}
\label{eq:FFRIR}
I_k = \displaystyle\sum_{z\in \mathcal{Z}_k}{G_z{R_z}^{-\alpha}}. 
\end{equation}
In the above expression, we have assumed that the nearest AP to the mobile in the $k^{th}$ tier is at a distance $r_k$, which is a random variable. Also the fading between the nearest AP in consideration is denoted by $g_k$.

With FFR, a mobile user first determines its $\sinr$ to the nearest AP of the $k^{th}$ tier and checks if it is less than the tier's FFR threshold $T_k$. If so, then the user is classified as an \textit{edge} user and the AP transmits its downlink on the reserved FFR band, randomly picked from $\Delta$ sub-bands available. Otherwise we classify the mobile as an \textit{interior} user. These classifications arise differently than prior work utilizing the typical grid model assumption which defines an interior radius \cite{Novlan2010}, since constant $\sinr$ contours can no longer be defined as concentric circles around the AP \cite{Hernandez2009}. In fact the edge or interior user classifications does not necessarily have the same geographic interpretation for each cell. As noted in \cite{Brown2000}, this consequence of the spatial PPP more closely reflects non-regular deployments and typically corresponds to a lower performance bound compared to the upper bound provided by the grid model.

To accommodate the difference between SFR and Strict FFR in terms of the use of power control, we introduce the design parameter $\beta$. Typical ranges for $\beta$ are 0-20 dB \cite{Doppler2009,Mazin2010}. Since this extra downlink power is only applied to $1/\Delta$ of the base stations on the first tier the interference power is given by $\eta P_1I_1 + \sum_{k=2}^{K}P_kI_k$, where $\eta = (\Delta - 1  + \beta)/\Delta$ consolidates the edge and interior downlinks into a single effective interference term.

\section{\label{sec:ClosedCoverage} Coverage Probability with Closed Access}
We initially consider coverage probability the downlink of a multi-tier network with closed access between the tiers. For example, in the context of a two-tier network with underlaid femtocells, a mobile user connected to the macrocell may be in range of a femtocell, but is unable to connect to that femtocell, potentially resulting in cross-tier interference.

Coverage probability is the probability that a user's $\sinr$ is greater than a threshold $T$,
\begin{equation}
\label{eq:coverage}
\bar{F}(T) = \PP(\sinr > T),   
\end{equation}
equivalently the CCDF of the $\sinr$ for a particular reuse strategy, denoted as $\bar{F}(T)$. 

\subsection{Single-tier coverage with FFR}
Our prior results in \cite{Novlan2011,NovlanGCOM2011} take advantage of the framework recently developed in \cite{Andrews2010} utilizing the Poisson point process (PPP) model for base station locations. The authors of \cite{Andrews2010} determine expressions for the exact distribution of the typical mobile's $\sinr$, with traditional per-cell frequency reuse for a single-tier of base stations. As a result, under reasonable assumptions for modern cellular networks, the results in \cite{Novlan2011} reduce to tractable expressions which provide insight into system design guidelines and the relative merits of Strict FFR and SFR, compared to universal reuse for a two-tier network with open access between tiers. Also in \cite{Novlan2011,NovlanGCOM2011}, the shape and values of the distributions derived for Strict FFR and SFR are shown to be closer to results obtained using location data from an actual base station deployment than simulations utilizing the standard grid model. We now provide the distribution of $\sinr$ for cell-edge users with Strict FFR and SFR under closed access.

\subsection{Multi-tier coverage with Strict FFR} 
In the case of Strict FFR, we assume that inter-cell and cross-tier interference is present on the common sub-band allocated to all macrocells, while the FFR sub-band is reserved for macrocell users and does not experience cross-tier interference, only inter-cell interference thinned with a reuse factor of $\Delta$. First tier edge users are those who have $\sinr$ less than the macrocell's FFR threshold $T_1$ on the common sub-band shared by all cells and are therefore selected by the reuse strategy to have a new sub-band allocated to them from the $\Delta$ total available sub-bands reserved for the edge users. 

\begin{theorem}[Strict FFR, closed access, edge user]
\label{thm:ClosedFFRedge}
The coverage probability of a first tier edge user in a strict FFR system, assigned a FFR sub-band is 
\begin{equation}
\bar{\mathrm{F}}_{\mathrm{FFR,c}}(T)=\frac{ \pi\lambda_1 \int_0^\infty e^{-\pi \lambda_1 v \left(1+\frac{\rho({T},\alpha)}{\Delta}\right)-\mu T \frac{\sigma^2}{P_1} v^{\alpha/2}} - e^{-\pi \lambda_1 v \left(1+2 \xi(T,{T_1},\alpha,\Delta)+2\sum_{k=2}^{K}\kappa_k\psi(\gamma_k{T_1},\alpha)\right)-\mu (T+{T_1}) \frac{\sigma^2}{P_1} v^{\alpha/2}} \dd v }{1-\pi \lambda_1 \int_0^\infty e^{-\pi \lambda_1 v \left(1+\rho({T_1},\alpha)+2\sum_{k=2}^{K}\kappa_k\psi(\gamma_k{T_1},\alpha)\right)} e^{-\mu (T+{T_1}) \frac{\sigma^2}{P_1} v^{\alpha/2}} \dd v},
\label{eq:FFRclosed}
\end{equation}
\begin{equation}
\label{eq:rho}
\textrm{where} ~\rho(z,\alpha) = z^{2/\alpha}\int_{z^{-2/\alpha}}^{\infty}\frac{1}{1+u^{\alpha/2}}\dd u,
\end{equation}
\begin{equation}
\label{eq:xi}
\xi(T,{T_1},\alpha,\Delta)= \int_{r_1}^\infty\left[1-\frac{1}{1+{T_1}{r_1}^\alpha x^{-\alpha}}\left(1-\frac{1}{\Delta}\left(1-\frac{1}{1+T {r_1}^\alpha x^{-\alpha}}\right) \right) \right] x \dd x,
\end{equation}
\begin{equation}
\label{eq:terms}
\textrm{and} ~ \psi(z,\alpha) = \textrm{csc}\left(\frac{2\pi}{\alpha}\right)\frac{\pi z^{2/\alpha}}{\alpha}~,~\gamma_k = \frac{P_k}{P_1}~,~\kappa_k = \frac{\lambda_k}{\lambda_1}.
\end{equation}
\end{theorem}

\begin{proof}
The proof is given in Appendix A.
\end{proof}

An immediate observation of this framework is that it leads to expressions which are only a function of the relevant FFR design parameters. The intra-tier interference before and after FFR is applied are captured in the $\xi(T,{T_1},\alpha,\Delta)$ and $\rho(z,\alpha)$ terms respectively, while the cross tier interference terms for each tier are expressed by $\psi(z,\alpha)$. 
%
%
\subsection{Multi-tier coverage with SFR}
We now consider the CCDF of the $\sinr$ for edge users with SFR. In this case all the subbands overlap with those of the other tiers since SFR makes use of the entire spectrum but allocates edge users with $\sinr$ below the FFR threshold a higher transmit power determined by the $\beta$ parameter.

\begin{theorem}[SFR, closed access, edge user]
\label{thm:ClosedSFRedge}
The coverage probability of an SFR edge user whose initial $\sinr$ is less than $T_1$ is 
\begin{eqnarray}\bar{\mathrm{F}}_{\mathrm{SFR,c}}(T)&=\frac{ \pi \lambda_1 \int_0^\infty e^{-\pi \lambda_1 v\left(1+\rho(\frac{\eta T}{\beta},\alpha)+2\sum_{k=2}^K \kappa_k \psi(\frac{\gamma_k}{\beta}T,\alpha)\right)} e^{-\mu (T) \frac{\sigma^2}{\beta P_1} v^{\alpha/2}} \dd v}{1-\pi \lambda_1 \int_0^\infty e^{-\pi \lambda_1 v\left(1+\rho(\eta{T_1},\alpha)+2\sum_{k=2}^K \kappa_k \psi(\gamma_k{T_1},\alpha)\right)} e^{-\mu (\eta{T_1}) \frac{\sigma^2}{P_1} v^{\alpha/2}} \dd v}\nonumber\\ 
&-\frac{\pi \lambda_1 \int_0^\infty e^{-\pi \lambda_1 v \left(1+2\zeta(T,{T_1},\alpha,\Delta,\beta,\eta)+2\sum_{k=2}^K \kappa_k\left(\psi(\frac{\gamma_k}{\beta}T,\alpha)+\psi(\gamma_k{T_1},\alpha)\right)\right)} e^{-\mu (T+\eta{T_1}) \frac{\sigma^2}{P_1} v^{\alpha/2}} \dd v }{1-\pi \lambda_1 \int_0^\infty e^{-\pi \lambda_1 v\left(1+\rho(\eta{T_1},\alpha)+2\sum_{k=2}\kappa_k \psi(\gamma_k{T_1},\alpha)\right)} e^{-\mu (\eta{T_1}) \frac{\sigma^2}{P_1} v^{\alpha/2}} \dd v}.
\label{eq:SFRclosed}
\end{eqnarray}
\[\textrm{where}~ \zeta(T,{T_1},\alpha,\Delta,\beta,\eta)= \int_{r_1}^\infty\left[1- \frac{1}{1+\eta {T_1}{r_1}^{\alpha} x^{-\alpha}}\frac{1}{1+\frac{\eta}{\beta}T{r_1}^{\alpha}x^{-\alpha}} \right] x \dd x,\]
$\rho(z,\alpha)$ is given by \eqref{eq:rho}, and $\psi(z,\alpha)$, $\kappa_k$ and $\gamma_k$ are given by \eqref{eq:terms}.
\end{theorem}

\begin{proof}
The proof is given in Appendix B.
\end{proof}
The expressions differ from Strict FFR both due to the effective $\sinr$ and FFR thresholds shaped by the power control factor $\beta$ and effective interference power $\eta$ respectively.

\subsection{Model Evaluation}
While all our coverage probability results hold for general pathloss exponents $\alpha$ and different noise powers $\sigma^2$, in this section we present a special case where $\alpha = 4$ and $\sigma^2 = 0$. For this case the coverage probability results reduce to simple closed-form expressions, allowing clear insight into the performance of cell-edge users, something not previously possible with the grid model. This choice of pathloss exponent is in the range of commonly used values in practice \cite{Rappaport2002}. Furthermore, most urban cellular networks - where FFR is of the most interest - are interference-limited and noise is negligible compared to the background interference from the adjacent BSs.

In the case of $\alpha=4$ and no noise, for \textit{Strict FFR}, the CCDF is given as, 
\begin{eqnarray}
\label{eq:FFRMTsimp}
\bar{\mathrm{F}}_\mathrm{FFR,e}(T) = \frac{1 + \rho(T_1)+\frac{\pi}{2}\sum_{k=2}^K\kappa_k \sqrt{\gamma_k T}}{\rho(T_1)+\frac{\pi}{2}\sum_{k=2}^K\kappa_k \sqrt{\gamma_k T}}\left(\frac{1}{1 + \frac{\rho(T)}{\Delta}} - \frac{1}{1 + 2\xi(T,T_1,\lambda,\Delta)+\frac{\pi}{2}\sum_{k=2}^K\kappa_k \sqrt{\gamma_k T}}\right),
\end{eqnarray}
\begin{equation}
\textrm{where}~ \xi(T,T_1,4,\Delta) = \frac{T\rho(T)-\rho(T_1)\left(T_1\Delta -T\left(1+\Delta\right)\right)}{4\Delta(T_1-T)},\textrm{and}~\rho(x) = \sqrt{x}\arctan{\left(\sqrt{x}\right)}. 
\end{equation}

In the case of $\alpha=4$ and no noise, for \textit{SFR}, the CCDF is given as, 
\begin{eqnarray}
\label{eq:SFRMTsimp}
\bar{\mathrm{F}}_\mathrm{SFR,c}(T) &=& \frac{1 + \rho(\eta T_1)+\frac{\pi}{2}\sum_{k=2}^K\sqrt{\gamma_k T}}{\rho(\eta T_1)+\frac{\pi}{2}\sum_{k=2}^K\sqrt{\gamma_k T}}\times\nonumber\\
&~&\left(\frac{1}{1 + \frac{\rho(\frac{\eta}{\beta} T)}{\Delta} + \frac{\pi}{2}\sum_{k=2}^K\sqrt{\frac{\gamma_k}{\beta} T}} - \frac{1}{1 + 2\zeta(T,T_1,\lambda,\Delta)+\frac{\pi}{2}\sum_{k=2}^K\sqrt{\gamma_k T}}\right),
\end{eqnarray}

\begin{eqnarray}
\textrm{where}~\zeta(T,T_1,\beta,\eta) &=& \frac{{\eta}^{3/2}T\beta}{4\sqrt{T_1}(T-T_1\beta)}-\frac{\eta\beta{T}^3\left(2\arctan\left(\sqrt{\frac{\beta}{\eta T}}\right)+\pi\right)}{(T-T_1\beta)}+\nonumber\\
&~&\frac{\eta {T}^{3/2}{T_1}^{3/2}\beta^{5/2}\left(2\arctan\left(\frac{1}{\sqrt{\eta T_1}}\right)-\pi\right)}{(T-T_1\beta)}.
\end{eqnarray}

Fig. \ref{fig:ClosedEdge} shows the derived distributions for Strict FFR and SFR edge users for a three-tier network with no noise and $\alpha = 4$ compared with Monte-Carlo simulations. The accuracy of the mathematical model is highlighted by the exact match of the curves with the simulation results. We also see the improved coverage afforded to cell-edge users with FFR compared to universal frequency reuse. For Strict FFR, much of the gain results in the removal of both cross-tier interference and $1/\Delta$ of the intra-tier interference. SFR provides a lower coverage gain, but this can be mitigated by the use of higher $\beta$, or taking into account that more spectrum is available than under Strict FFR since each cell fully utilizes all subbands.

Using similar techniques we can derive the distributions for interior macro or femto users using this framework. Additionally, these results are also valid for $\alpha \neq 4$, but the expressions no longer have the same simple closed-form. Instead they are integrals that can be evaluated using numerical techniques.

\section{\label{sec:OpenCoverage} Coverage Probability with Open Access}
In the following analysis of open access downlinks we make the following two assumptions, (i) that there are only two-tiers of access points, and (ii) we only consider the $\sir$, as the access metric, neglecting noise. While our general framework can accommodate an unlimited number of tiers and noise, making those assumptions greatly reduces the complexity of the expressions for the $\sir$ distributions. The following $\sir$ distributions for Strict FFR and SFR are a function of two open access thresholds, $T_1$ set by the macro tier and $T_2$ set by the second tier of APs. The open access thresholds determine whether a user is switched to a reuse-$\Delta$ sub-band or served by a either the common band of the macrocell or the nearest second-tier AP. 

Let $\sir_1$ and $\sir_2$ denote the $\sir$ at the typical mobile of the closest first and second tier AP respectively, 
\begin{equation}
\label{eq:SIR1}
\sir_1 = \frac{P_1{g_1}{r_1}^{-\alpha}}{P_1I_1+P_2I_2+P_2{g_2}{r_2}^{-\alpha}}~,~\sir_2 = \frac{P_2{g_2}{r_2}^{-\alpha}}{P_1I_1+P_2I_2+P_1{g_1}{r_1}^{-\alpha}}.
\end{equation}
Here $r_1$ denotes the distance of the mobile at the origin to the nearest macro BS, and $r_2$ the distance to the nearest femtocell. The interference caused by the macro BSs is denoted by $I_1$, while $I_2$ is the interference caused by the femtocells, excluding the closest one. If for a mobile user, $\sir_1 < T_1$ and $\sir_2 < T_2$, then the mobile user is allocated a new FFR sub-band $\delta_y$, where $\delta \in \{1,...,\Delta\}$ with uniform probability $\frac{1}{\Delta}$ and a new $\sir$ given by $\hat{\sir}$ which is different under Strict FFR or SFR. 
The CCDF of the edge user $\sir$ under open access is given by
\begin{equation}
\label{eq:pcCond}
\bar{\mathrm{F}}_{\mathrm{FFR,open,e}}(T) = \mathbb{P}\left(\hat{\sir} > T ~|~ \sir_1 < T_1 ~,~ \sir_2 < T_2\right).
\end{equation}
As we can see from \eqref{eq:pcCond} the analysis of the coverage probability is more complicated relative to closed access due to the inter-dependence of the terms $\sir_1$ and $\sir_2$.

\subsection{Strict FFR} 
First we consider the distribution of \eqref{eq:pcCond} for Strict FFR. Since the mobile user is allocated a different sub-band, it experiences new fading power $\hat{g_1}$ and out-of-cell interference $P_1\hat{I_1}$, which does not have cross-tier interference. 

\begin{theorem}[Strict FFR, open access, edge user]
\label{thm:OpenFFRedge}
The coverage probability of an edge user in a strict FFR system, assigned a FFR sub-band is 
\begin{eqnarray}
\label{eq:OpenBayesStrictFFR}
\bar{\mathrm{F}}_{\mathrm{FFR,o}}(T)&=&\frac{ p_{c}(T,\lambda_1,\alpha,\Delta)-\int_{0}^{\infty}\int_{0}^{\infty}\left(2\pi\lambda_1r_1 e^{-\pi\lambda_1{r_1}^2}\right)\left(2\pi\lambda_2r_2e^{-\pi\lambda_2{r_2}^2}\right) g_n(r_1,r_2) \dd r_1 \dd r_2 }{\int_{0}^{\infty}\int_{0}^{\infty}\left(2\pi\lambda_1r_1 e^{-\pi\lambda_1{r_1}^2}\right)\left(2\pi\lambda_2r_2e^{-\pi\lambda_2{r_2}^2}\right) g_d(r_1,r_2) \dd r_1 \dd r_2}\nonumber
\end{eqnarray}
\begin{align*}
\textrm{where}~ g_d(r_1,r_2) &=
1-\epsilon_1e^{\left(-2\pi\lambda_1\rho_{1,1}(T_1,\alpha)\right)} e^{\left(-2\pi\lambda_2\rho_{1,2}(\gamma T_1,\alpha)\right)}-\epsilon_2e^{\left(-2\pi\lambda_1\rho_{2,1}(T_2/\gamma,\alpha)\right)} e^{\left(-2\pi\lambda_2\rho_{2,2}(T_2,\alpha)\right)},\\
g_n(r_1,r_2)
&= \epsilon_1e^{-2\pi \left( \lambda_1 \xi_{1,1}\left(T,T_1,\alpha,\Delta\right)+ \lambda_2 \rho_{1,2}\left(T_1,\alpha\right)\right)} + \epsilon_2e^{-2\pi \left(\lambda_1 \xi_{2,1}\left(T,T_2/\gamma,\alpha,\Delta\right) + \lambda_2 \rho_{2,2}\left(T_2,\alpha\right)\right)},
\end{align*}
\begin{equation}
\label{eq:xiClosed}
\xi_{a,b}\left(T,z,\alpha,\Delta\right) = \int_{r_b}^\infty\left[1- \frac{1}{1 + z r_a^{\alpha} x^{-\alpha}}\left(1-\frac{1}{\Delta}\left(1-\frac{1}{1 + T r_b^{\alpha} x^{-\alpha}}\right)\right) \right] x \dd x,
\end{equation}
\begin{equation}
\label{eq:rhoOpen}
\rho_{a,b}\left(z,\alpha\right) = \int_{r_b}^{\infty}\left(1-\frac{1}{1 + z r_a^{\alpha}x^{-\alpha}}\right)x dx,
\end{equation}
\begin{equation}
\label{eq:OpenTerms}
\textrm{and}~\gamma = \frac{P_2}{P_1} ~,~ \epsilon_1 = \left(\frac{1}{T_1\gamma\frac{{r_1}^{\alpha}}{{r_2}^{\alpha}}+1}\right), ~\textrm{and}~ \epsilon_2 = \left(\frac{1}{T_2\left(\gamma \frac{{r_1}^{\alpha}}{{r_2}^{\alpha}}\right)^{-1}+1}\right).
\end{equation}
\end{theorem}

\begin{proof}
The proof is given in Appendix C.
\end{proof}
Compared to the closed access results, the derivations are not nearly as clean due to the dependence of the user's $\sir$ on $r_1$ and $r_2$. The derivations require evaluating a double integral which does not have a closed form. In fact, the number of tiers under consideration determines the number of integrals which must be evaluated. Despite this, we can still obtain insight into the underlying nature of the distributions. Also, it is expected that most practical deployments would not have more than about three tiers even in dense environments, making this analysis practical through the use of numerical evaluation of the integrals.

\subsection{SFR}
As was the case for closed access, the SFR expressions differ from Strict FFR due to the power control factor and effective interference power. Additionally the full $\Delta$-reuse of subbands with SFR results in cross-tier interference for the edge users as well as interior users. We now give the expression for coverage probability with open access and SFR based on the $\sir$ in \eqref{eq:pcCond}.

\begin{theorem}[SFR, open access, edge user]
\label{thm:OpenSFRedge}
The coverage probability of an SFR edge user whose initial $\sir$ is less than $T_1$ and $T_2$ is 
\begin{eqnarray}
\label{eq:SFRopenFinal}
\bar{\mathrm{F}}_{\mathrm{SFR,o}}(T)&=&\frac{\pi \lambda_1 \int_0^\infty e^{-\pi \lambda_1 v\left(1+\rho(\frac{\eta}{\beta} T,\alpha)+2\kappa\psi(\frac{\gamma}{\beta}T,\alpha)\right)} \dd v }{\int_{0}^{\infty}\int_{0}^{\infty}\left(2\pi\lambda_1r_1 e^{-\pi\lambda_1{r_1}^2}\right)\left(2\pi\lambda_2r_2e^{-\pi\lambda_2{r_2}^2}\right) f_d(r_1,r_2)\dd r_1 \dd r_2}\nonumber\\ 
&-&\frac{\int_{0}^{\infty}\int_{0}^{\infty}\left(2\pi\lambda_1r_1 e^{-\pi\lambda_1{r_1}^2}\right)\left(2\pi\lambda_2r_2e^{-\pi\lambda_2{r_2}^2}\right) f_n(r_1,r_2) \dd r_1 \dd r_2 }{\int_{0}^{\infty}\int_{0}^{\infty}\left(2\pi\lambda_1r_1 e^{-\pi\lambda_1{r_1}^2}\right)\left(2\pi\lambda_2r_2e^{-\pi\lambda_2{r_2}^2}\right) f_d(r_1,r_2) \dd r_1 \dd r_2}.\nonumber
\end{eqnarray}
\begin{align*}
\textrm{where}~f_n(r_1,r_2)&= \epsilon_1e^{-2\pi\lambda_1 \left(\zeta_{1,1}\left(T,T_1,\alpha,\Delta,\beta,\eta\right)+\kappa\psi\left(\frac{\gamma}{\beta}T,\alpha\right)+ \kappa \rho_{1,2}\left(\gamma T_1,\alpha\right)\right)}\\
&+ \epsilon_2e^{-2\pi \lambda_1 \left(\zeta_{2,1}\left(T,T_2/\gamma,\alpha,\Delta,\beta,\eta\right)+\kappa \psi\left(\frac{\gamma}{\beta}T,\alpha\right)+\kappa \rho_{2,2}\left(T_2,\alpha\right)\right)},
\end{align*}
\begin{align*}
f_d(r_1,r_2)=1-\epsilon_1e^{\left(-2\pi\lambda_1\left(\rho_{1,1}(\eta T_1,\alpha)+ \kappa\rho_{1,2}(\gamma T_1,\alpha)\right)\right)}-\epsilon_2e^{\left(-2\pi\lambda_1\left(\rho_{2,1}\left(\frac{\eta}{\gamma}T_2,\alpha\right)+\kappa\rho_{2,2}\left(T_2,\alpha\right)\right)\right)},
\end{align*}
\begin{equation}
\label{eq:zetaOpen}
\zeta_{a,b}(y,z,\beta,\eta)= \frac{1}{2(y-z)}\left(y\rho_{a,b}(y,\alpha)+z\rho_{a,b}(z,\alpha)\right),\textrm{and}~\rho_{a,b}\left(z,\alpha\right)~\textrm{given by \eqref{eq:rhoOpen}}.
\end{equation}
\end{theorem}

\begin{proof}
The proof is given in Appendix D.
\end{proof}
The expressions have a similar form but differ from Strict FFR due to the effect of $\eta$ and $\beta$ on the $\sir$ and FFR thresholds. As with Strict FFR, the derivations do not reduce as simply as closed access expressions due to the dependence of the user's $\sir$ on $r_1$ and $r_2$, but still can be computed with a single integral in the case of $\sigma^2 = 0$ and $\alpha = 4$.

\subsection{Model Evaluation}
Fig. \ref{fig:OpenEdge} shows the derived distributions for Strict FFR and SFR edge users for a two-tier network with no noise and $\alpha = 4$ compared with Monte-Carlo simulations. As with closed access, the curves match exactly. We also note that there is an upwards shift in the coverage probability curves, due to the impact of off-loading of users onto the secondary tier. With closed access, users whose $\sinr$ falls below the first tier FFR threshold $T_1 = 1$dB would be assigned a FFR band and may or may not be able to be covered due to interference or propagation challenges, however if their $\sinr$ to a second tier AP is greater than $T_2 = 5$dB, they are guaranteed coverage and affect the distribution of the users who utilize FFR. The selection of the FFR thresholds is further investigated in the following section.

\section{System Design Implications \label{sec:applications}}
In this section we present several applications of the Strict FFR and SFR $\sinr$ and $\sir$ distributions derived for closed and open access in Sections \ref{sec:ClosedCoverage} and \ref{sec:OpenCoverage}, which illustrate how they can be used to provide additional tools and insight for the system design of heterogeneous networks utilizing FFR. 

\subsection{Average Edge User Rate}
In modern cellular networks, the important metric of average achievable rate can be derived from the $\sinr$ statistics. In this section we illustrate how the coverage results derived in Section \ref{sec:ClosedCoverage} and \ref{sec:OpenCoverage} can be straightforwardly extended to develop average edge user rate expressions under Strict FFR or SFR.  

The average data rate $\bar{\tau} = \mathbb{E}\left[\ln\left(1 + \sinr\right)\right]$ is achieved by the users, assuming adaptive modulation and coding, and the expressions are given in terms of nats/Hz, where 1 bit $= \log_e(2)$ nats. The average rate of an edge user is determined by integrating over the $\sinr$ distribution and fading. Due to the two-stage nature of FFR the $\sinr_{\textrm{e}}$ of the edge user on the new subband is conditioned on the previous $\sinr_{\textrm{i}}$ on the common subband. Thus we have
\begin{eqnarray}
\bar{\tau} &=& \mathbb{E}\left[\ln\left(1 + \sinr\right)\right] = \int_{r > 0}e^{-\pi\lambda r^2}\mathbb{E}\left[\ln\left(1 + \sinr_{\textrm{e}}\right)\right]2\pi\lambda r\dd r, \nonumber\\
&=& \int_{r > 0}e^{-\pi\lambda r^2}\int_{t > 0} \mathbb{P}\left[\ln\left(1 + \sinr_{\textrm{e}}\right) > t ~\bigg|~ \sinr_{\textrm{i}} < \TR\right]2\pi\lambda \dd t ~ r\dd r.\nonumber
\end{eqnarray}
where we use the fact that since the rate $\tau = \ln(1 + \sinr)$ is a positive random variable, $E[\tau] = \int_{t>0} P(\tau > t)\dd t$. From the above expression we see that the derivation of these terms involves substituting $e^{t-1}$ in place of the $\sinr$ threshold $T$ and computing an additional integral.

\subsection{Multi-tier interference and closed access}
We now consider a two-tier network with Strict FFR for the macro users and closed access and show the connection between the density ratio of the tiers $\kappa$ and the $\sinr$ distribution. Fig. \ref{fig:MTkappa} plots the distribution for edge users as an increasing function of $\kappa$, effectively increasing the density of second-tier APs. As $\kappa$ increases we see in Fig. \ref{fig:MTkappa} that the $\sinr$ increases for macrocell users. This  is a consequence of the use of Strict FFR, since the FFR bands are reserved for only macrocell users, any user moving from the common band to the FFR band will see a reduction in interference. As the interference from the second tier increases with $\kappa$, more and more macro users have $\sinr$ below $T_1$ and since they cannot connect to the second tier due to the closed access constraint, they must be moved onto a FFR sub-band. The implication of this result is that the size of the partitions will need to be increased, which for Strict FFR, can cause the overall sum rate of the macrocells to decrease due to the reduction in overall spectrum usage. 

\subsection{Open access FFR thresholds}
In Fig. \ref{fig:MTt2} the SFR edge user $\sir$ CCDF is shown for different values of $T_2$, the second-tier FFR threshold under open access. Decreasing $T_2$ increases the number of mobile users which can connect to that AP on the common sub-band. From Fig. \ref{fig:MTt2} we see that this results in the overall increase of the $\sir$ of the edge users. In other words, as $T_2$ increases, only the users with the worst $\sir$ are given FFR sub-bands and they also are the users who can have the greatest benefit from the FFR sub-bands.

A related concept is called biasing, in which the access thresholds of the femtocells or other secondary APs are adjusted in order to increase the offload from the macrocell. The reasons for biasing may not be solely related to the ability of the macrocell to provide coverage for a user, but rather to reduce traffic for especially overloaded macrocells. Our proposed framework can implicitly capture this effect in the design of $T_1$ and $T_2$. By raising $T_1$ and lowering $T_2$ we can define a middle $\sir$ range $T_{Bias} = T_1-T_2$, wherein a desired percentage of macrocell users are offloaded.    

\section{\label{sec:conclusion} Conclusion}
This work has presented a new tractable analytical framework for evaluating coverage probability in heterogeneous networks utilizing Strict FFR and SFR which captures the non-uniformity of these deployments and gives insight into the performance tradeoffs of those FFR strategies. The model presented in this work can be utilized as a foundation for future research in interference management and performance analysis of heterogeneous networks utilizing dynamic FFR strategies for addressing changing channel conditions and user traffic in the network \cite{Ali2009,Gustavo2009,Stolyar2009}. Additionally, in the uplink, the constraints of power control, mobility of the interfering mobiles, and the important metric of power consumption at the mobile device impact the system design, make analysis very challenging using the traditional grid model \cite{Wamser2010}. Tractable analysis should assist system designers in evaluating the performance of potential algorithms in non-uniform and multi-tier deployments. 

\appendices
\section{\label{sec:AppA} Proof of Strict FFR, Closed Access Theorem}
A macrocell connected user $y$ with $\sinr_1 < {T_1}$ is given a FFR sub-band $\delta_y$, where $\delta \in \{1,...,\Delta\}$ with uniform probability $\frac{1}{\Delta}$, and experiences new fading power $\hat{g_1}$ and out-of-cell interference $P_1\hat{I_1}$, instead of $g_1$ and $P_1I_1 + \displaystyle\sum_{k=2}^K P_kI_k$. The CCDF of the edge user $\bar{\mathrm{F}}_{\mathrm{FFR,c}}(T)$ is now conditioned on its previous $\sinr$. Using Bayes' rule we have,
\begin{equation}
\mathbb{P}\left(\frac{P_1\hat{g_1} {r_1}^{-\alpha}}{\sigma^2+P_1\hat{I_1}} > T ~\bigg|~ \frac{P_1g_1{r_1}^{-\alpha}}{\sigma^2+P_1I_1+\sum_{k=2}^K P_kI_k} < {T_1} \right)=
\frac{\mathbb{P}\left(\frac{P_1\hat{g_1} {r_1}^{-\alpha}}{\sigma^2+P_1\hat{I_1}} > T \ , \frac{P_1g_1{r_1}^{-\alpha}}{\sigma^2+P_1I_1+\sum_{k=2}^K P_kI_k} < {T_1} \right)}{\mathbb{P}\left(\frac{P_1g_1{r_1}^{-\alpha}}{\sigma^2+P_1I_1+\sum_{k=2}^K P_kI_k} < {T_1} \right)}.
\label{eq:PcClosedFFR}
\end{equation}

Conditioning on $r_1$, the distance to the nearest BS, which is Rayleigh distributed and focusing on the numerator of \eqref{eq:PcClosedFFR}, since $\hat{g_1}$ and $g_1$ are i.i.d. exponentially distributed with mean $\mu$, gives
\[ \E \left[e^{\left(-\mu \frac{T}{P_1} {r_1}^{\alpha}(\sigma^2 + P_1\hat{I_1}) \right)}\right] - \E \left[e^{\left(-\mu \frac{T}{P_1} {r_1}^{\alpha}(\sigma^2 + P_1\hat{I_1}) \right)}e^{\left(-\mu \frac{{T_1}}{P_1} {r_1}^{\alpha}\left(\sigma^2 + P_1I_1 +\sum_{k=2}^K P_kI_k \right) \right)}\right],\] 

Factoring out terms dependent on the independent noise power $\sigma^2$ we observe that the expectation of the second term with respect to $\hat{I_1}$, $I_1$, $I_2$, ... , and $I_K$ is the joint Laplace transform $\mathscr{L}(\hat{s_1},s_1,s_2,...,s_{K})$ of $\hat{I_1}$, $I_1$, $I_2$, ..., and $I_K$ given by 
\begin{align*}
&= \E \left[\exp \left(-\hat{s_1}\hat{I_1}-s_1 I_1 - \sum_{k=2}^K s_k I_{k} \right)\right]\\
&= \E \left[\exp \left(-\hat{s_1}\sum_{z \in Z_1} \hat{G_z}{R_{z}}^{-\alpha}\mathbf{1}(\delta_z = \delta_y) - s_1\sum_{z \in Z_1}G_z {R_{z}}^{-\alpha} - \sum_{k=2}^K \left(s_k\sum_{z \in Z_k}G_z {R_{z}}^{-\alpha}\right)\right)\right]\\
&= \E \left[\prod_{z \in Z_1}\left(1-\E\left[\mathbf{1}(\delta_z = \delta_y)\right](1-e^{(-\hat{s_1}\hat{G_z}{R_{z}}^{-\alpha})})\right)e^{\left(-s_1 G_z{R_{z}}^{-\alpha}\right)}\right] \prod_{k=2}^K\E\left[\prod_{z\in Z_k}e^{\left(-s_k G_z{R_{z}}^{-\alpha}\right)}\right],
\end{align*}
where $\mathbf{1}(\delta_y = \delta_z)$ is an indicator function that takes the value 1 if base station $z$ is transmitting to an edge user on the same sub-band $\delta$ as user $y$, and the third step arises from the independence of $I_1$ and $\hat{I_1}$ with respect to $I_2$,...,$I_K$. Since $\hat{G_z}$ and $G_z$ are also exponential random variables with mean $\mu$, we can evaluate the above expression as
\begin{equation}
\E \left[\prod_{z \in Z_1}\left(1-\frac{1}{\Delta}\left(1-\frac{\mu}{\mu+\hat{s_1}{R_{z}}^{-\alpha}}\right) \right)\frac{\mu}{\mu+s_1 {R_{z}}^{-\alpha}}\right]\prod_{k=2}^{K}\E \left[\prod_{z \in Z_k}\frac{\mu}{\mu+s_k {R_{z}}^{-\alpha}}\right].\nonumber
\label{eq:closedLFFR}
\end{equation}
By using the probability generating functional (PGFL) of the PPP \cite{Stoyan1996} we obtain
\[\mathscr{L}(\hat{s_1},s_1,s_2,...,s_K)= e^{\left(-2\pi \lambda_1 \int_{r_1}^\infty\left[1- \frac{\mu}{\mu+s_1 x^{-\alpha}}\left(1-\frac{1}{\Delta}\left(1-\frac{\mu}{\mu+\hat{s_1} x^{-\alpha}}\right) \right) \right] x \dd x \right)}\prod_{k=2}^{K}e^{\left(-2\pi\lambda_k{\left(\frac{s_k}{\mu}\right)}^{2/\alpha}\frac{\pi\rm{csc}(\frac{2\pi}{\alpha})}{\alpha}\right)}.\]

Substituting for the integration variables $s$ and de-conditioning on $r_1$, we have 
\begin{eqnarray}
\label{eq:ClosedFFRnum}
2 \pi r_1\lambda_1 \int_0^\infty e^{-\pi \lambda_1 {r_1}^2 \left(1+2\xi(T,{T_1},\alpha,\Delta)+2\sum_{k=2}^{K}\kappa_k\psi(\gamma_k{T_1},\alpha)\right)-\mu (T+{T_1}) \frac{\sigma^2}{P_1} {r_1}^{\alpha}} \dd r_1, 
\end{eqnarray}
\[\textrm{where}~\xi(T,{T_1},\alpha,\Delta)= \int_{r_1}^\infty\left[1-\frac{1}{1+{T_1}{r_1}^\alpha x^{-\alpha}}\left(1-\frac{1}{\Delta}\left(1-\frac{1}{1+T {r_1}^\alpha x^{-\alpha}}\right) \right) \right] x \dd x,\]
\[\textrm{and}~\psi(z,\alpha) = \rm{csc}\left(\frac{2\pi}{\alpha}\right)\frac{\pi {z}^{2/\alpha}}{\alpha}~,~\gamma_k = \frac{P_k}{P_1}~,~\kappa_k = \frac{\lambda_k}{\lambda_1}.\]

Now we focus on the denominator of \eqref{eq:PcClosedFFR}, using the independence of $I_1$ and $I_2$,...,$I_K$ we have
\begin{eqnarray}
&~&1-\E \left[\exp\left(-\mu \frac{{T_1}}{P_1} {r_1}^{\alpha}(\sigma^2 + P_1I_1 + \displaystyle\sum_{k=2}^K P_kI_k) \right)\right]\nonumber\\&=& 1-\E \left[e^{\left(-\mu \frac{{T_1}}{P_1} {r_1}^{\alpha}(\sigma^2 + P_1I_1) \right)}\right]\prod_{k=2}^{K}\E \left[e^{\left(-\mu \frac{{T_1}}{P_1} {r_1}^{\alpha}(P_kI_k) \right)}\right]\nonumber\\
&=& 1-2\pi r_1 \lambda_1 \int_0^\infty e^{-\pi\lambda_1 {r_1}^2\left(1+\rho({T_1},\alpha)+2\sum_{k=2}^{K}\kappa_k\psi(\gamma_k{T_1},\alpha)\right)} e^{-\mu (T+{T_1}) \frac{\sigma^2}{P_1} {r_1}^{\alpha}} \dd r_1.
\label{eq:ClosedFFRdenom}
\end{eqnarray}

The first term of the numerator represents the $\sinr$ on the newly allocated subband we have 
\begin{equation}
\label{eq:firstClosedFFR}
\pi\lambda_1 \int_0^\infty e^{-\pi \lambda_1 v \left(1+\frac{\rho({T},\alpha)}{\Delta}\right)-\mu T \frac{\sigma^2}{P_1} v^{\alpha/2}},
\end{equation}
since the received interference is only from the first tier APs due to the closed access frequency allocation for edge users and is originally given in \cite{Andrews2010}.

Thus plugging \eqref{eq:ClosedFFRnum} and \eqref{eq:ClosedFFRdenom} back into \eqref{eq:PcClosedFFR}, and substituting \eqref{eq:firstClosedFFR} for the first term of the numerator and substituting ${r_1}^2 = v$ we have \eqref{eq:FFRclosed}.

\section{\label{sec:AppB} Proof of SFR, Closed Access Theorem}
A macrocell connected user $y$ with $\sinr < {T_1}$ is assigned a FFR sub-band $\delta_y$, where $\delta \in \{1,...,\Delta\}$ with uniform probability $\frac{1}{\Delta}$, and experiences new fading power $\hat{g_1}$, transmit power $\beta P_1$, and out-of-cell interference. The CCDF of the edge user $\bar{\mathrm{F}}_{\mathrm{SFR,c}}(T)$ is now conditioned on its previous $\sinr$,
\begin{equation}
\label{eq:PcClosedSFR}
\bar{\mathrm{F}}_{\mathrm{SFR,e}}(T)=\mathbb{P}\left(\frac{\beta P_1\hat{g_1} {r_1}^{-\alpha}}{\sigma^2+\eta P_1\hat{I_1} + \sum_{k=2}^K P_k\hat{I_k}} > T ~\bigg|~ \frac{P_1g_1{r_1}^{-\alpha}}{\sigma^2+ \eta P_1I_1+\sum_{k=2}^K P_kI_k} < {T_1} \right).
\end{equation}
Using Bayes' rule as in Theorem \ref{thm:ClosedFFRedge} and focusing on the resulting numerator, since $\hat{g_1}$ and $g_1$ are i.i.d. exponentially distributed with mean $\mu$, this gives
\[ \E \left[e^{\left(-\mu \frac{T}{\beta P_1} {r_1}^{\alpha}(\sigma^2 + \eta P_1\hat{I_1} + \sum_{k=2}^K P_k\hat{I_k}) \right)}\right] - \E \left[e^{\left(-\mu \frac{T}{\beta P_1} {r_1}^{\alpha}(\sigma^2 + \eta P_1\hat{I_1} + \sum_{k=2}^K P_k\hat{I_k}) \right)}e^{\left(-\mu \frac{{T_1}}{P_1} {r_1}^{\alpha}(\sigma^2 + \eta P_1I_1 + \sum_{k=2}^K P_kI_k) \right)}\right],\] 

Now concentrating on the second term, factoring out terms corresponding to the independent noise power $\sigma^2$, and {\em conditioning on $r_1$}, we obtain the joint Laplace transform of 
$\hat{I_1}$, $\hat{I_2}$, ..., $\hat{I_K}$, and $I_1$, $I_2$, ..., $I_K$ given by 
\begin{equation}
\E \left[\prod_{z \in Z_1}\frac{\mu}{\mu+\hat{s_1}{R_{z}}^{-\alpha}}\frac{\mu}{\mu+s_1 {R_{z}}^{-\alpha}}\right]\prod_{k=2}^K\E \left[\prod_{z \in \hat{Z_k}}\frac{\mu}{\mu+\hat{s_k} {R_{z}}^{-\alpha}}\right]\E \left[\prod_{z \in Z_k}\frac{\mu}{\mu+s_k {R_{z}}^{-\alpha}}\right].
\label{eq:closedLSFR}
\end{equation}
%
Using the same method as Theorem \ref{thm:ClosedFFRedge} and de-conditioning on $r_1$ we obtain
\begin{eqnarray}
\label{eq:SFRclosed2ndTerm}
2\pi r_1\lambda_1 \int_0^\infty e^{-\pi \lambda_1 {r_1}^2\left(1+2\zeta(T,{T_1},\alpha,\Delta,\beta,\eta)+2\sum_{k=2}^K\kappa_k\left(\psi(\frac{\gamma_k}{\beta}T,\alpha)+\psi(\gamma_k{T_1},\alpha)\right)\right)-\mu (\frac{T}{\beta}+{T_1}) \frac{\sigma^2}{P_1} {r_1}^{\alpha}} \dd r_1 \nonumber, 
\end{eqnarray}
\[\textrm{where}~\zeta(T,{T_1},\alpha,\Delta,\beta,\eta)= \int_{r_1}^\infty\left[1- \frac{1}{1+\eta {T_1}{r_1}^{\alpha} x^{-\alpha}}\frac{1}{1+\frac{\eta}{\beta}T{r_1}^{\alpha}x^{-\alpha}} \right] x \dd x,\]
Using the same argument and analysis for the resulting denominator of \eqref{eq:PcClosedSFR} after Bayes' rule is applied we have
\begin{equation}
\label{eq:denomClosedSFR}
1-2\pi r_1 \lambda_1 \int_0^\infty e^{-\pi \lambda_1 {r_1}^2\left(1+\rho(\eta{T_1},\alpha)+2\sum_{k=2}^K\kappa_k\psi({T_1},\alpha)\right)} e^{-\mu (T+\eta{T_1}) \frac{\sigma^2}{P_1} {r_1}^{\alpha}} \dd r_1,
\end{equation}
Finally, the first term of the numerator is given as
\begin{equation}
\label{eq:numSFRclosed}
2\pi r_1\lambda_1 \int_0^\infty e^{-\pi \lambda_1 {r_1}^2\left(1+\rho(\frac{\eta}{\beta} T,\alpha)+2\sum_{k=2}^K\kappa_k\psi(\frac{\gamma_k}{\beta}T,\alpha)\right)} e^{-\mu (T) \frac{\sigma^2}{\beta P_1} {r_1}^{\alpha}} \dd r_1.
\end{equation}
Thus plugging \eqref{eq:SFRclosed2ndTerm}, \eqref{eq:denomClosedSFR}, and \eqref{eq:numSFRclosed} back into \eqref{eq:PcClosedSFR} and substituting ${r_1}^2 = v$ we have \eqref{eq:SFRclosed}.

\section{\label{sec:AppC} Proof of Strict FFR, Open Access Theorem}
A user $y$ with $\sir_1 < T_1$ when connected to the closes macrocell and $\sir_2 < T_2$ when connected to the closest microcell is given a FFR sub-band $\delta_y$, where $\delta \in \{1,...,\Delta\}$ with uniform probability $\frac{1}{\Delta}$, and experiences new fading power $\hat{g_1}$ and out-of-cell interference $P_1\hat{I_1}$. The CCDF of the edge user $\bar{\mathrm{F}}_{\mathrm{FFR,o}}(T)$ is now conditioned on its previous $\sir$ and $r_1$ and $r_2$, the distance to the nearest tier 1 and tier 2 AP respectively, given by
\begin{eqnarray}
&~&\mathbb{P}\left(\frac{P_1\hat{{g_1}} {r_1}^{-\alpha}}{P_1\hat{I_1}} > T ~\bigg|~ \frac{P_1{g_1}{r_1}^{-\alpha}}{P_1I_1+P_2I_2+P_2{g_2}{r_2}^{-\alpha}} < T_1 , \frac{P_2{g_2}{r_2}^{-\alpha}}{P_1I_1+P_2I_2 + P_1{g_1}{r_1}^{-\alpha}} < T_2 \right).
\label{eq:OpenStrictFFRBayes}
\end{eqnarray}
Using Bayes' rule and initially focusing on the denominator, the conditional term in \eqref{eq:OpenStrictFFRBayes}, and conditioning on $g_2$ gives
\[\mathbb{P}\left(\frac{{r_1}^{\alpha}}{P_1}\left(\frac{P_2}{T_2}{g_2}{r_2}^{-\alpha}-\left(P_1I_1+P_2I_2\right)\right) ~<~ {g_1} ~<~ 
T_1\frac{{r_1}^{\alpha}}{P_1}\left(P_1I_1+P_2I_2+P_2{g_2}{r_2}^{-\alpha}\right) ~\bigg|~ {g_2} \right)\mathbb{P}\left({g_2}\right).\]
Since ${g_1}$ and ${g_2}$ are i.i.d. exponentially distributed with mean $\mu$, and setting $\bar{I} = P_1I_1 + P_2I_2$, this gives
\begin{align*}
&~ \E_{{g_2}} \left[\int_{\frac{{r_1}^{\alpha}}{P_1}\left(\frac{P_2}{T_2}{g_2}{r_2}^{-\alpha}-\bar{I}\right)^{+}}^{T_1\frac{{r_1}^{\alpha}}{P_1}\left(\bar{I}+P_2{g_2}{r_2}^{-\alpha}\right)}\mu e^{-\mu x}dx\right] = \E_{{g_2}} \left[e^{-\mu\frac{{r_1}^{\alpha}}{P_1}\left(\frac{P_2}{T_2}{g_2}{r_2}^{-\alpha}-\bar{I}\right)^{+}}-e^{-\mu T_1\frac{{r_1}^{\alpha}}{P_1}\left(\bar{I}+P_2{g_2}{r_2}^{-\alpha}\right)}\right],
\end{align*}
\begin{displaymath}
   \textrm{where}~(x)^{+} = \left\{
     \begin{array}{lr}
       x & : x > 0\\
       0 & : x < 0
     \end{array}
   \right.
\end{displaymath}
Evaluating the expectation, collecting terms and simplifying gives,
\begin{eqnarray}
&=&1-\epsilon_1e^{-\bar{I}\mu T_1\frac{{r_1}^{\alpha}}{P_1}} -\epsilon_2e^{-\bar{I}\mu T_2\frac{{r_2}^{\alpha}}{P_2}},~\textrm{where}
\label{eq:denomsimpFFRclosed}
\end{eqnarray}
\[\gamma = \frac{P_2}{P_1} ~,~ \epsilon_1 = \left(\frac{1}{T_1\gamma\frac{{r_1}^{\alpha}}{{r_2}^{\alpha}}+1}\right), ~\textrm{and}~ \epsilon_2 = \left(\frac{1}{T_2\left(\gamma \frac{{r_1}^{\alpha}}{{r_2}^{\alpha}}\right)^{-1}+1}\right).\]

We observe that the expectation of \eqref{eq:denomsimpFFRclosed} with respect to $I_1$ and $I_2$ is the joint Laplace transform of $I_1$ and $I_2$ evaluated at $(\mu T_1 \frac{{r_1}^{\alpha}}{P_1}, \mu  T_2 \frac{{r_2}^{\alpha}}{P_2})$. The joint Laplace transform denoted by $g_d(r_1,r_2)$ is
\begin{align*}
g_d(r_1,r_2) &= \E_{I_1,I_2} \left[1-\epsilon_1 e^{-s_1\bar{I}}-\epsilon_2 e^{-s_2\bar{I}}\right]\\
&=1-\epsilon_1e^{\left(-2\pi\lambda_1\rho_{1,1}(T_1,\alpha)\right)} e^{\left(-2\pi\lambda_2\rho_{1,2}(\gamma T_1,\alpha)\right)}-\epsilon_2e^{\left(-2\pi\lambda_1\rho_{2,1}(T_2/\gamma,\alpha)\right)} e^{\left(-2\pi\lambda_2\rho_{2,2}(T_2,\alpha)\right)},
\end{align*}
where $\rho_{a,b}\left(z,\alpha\right)$ is given by \eqref{eq:rhoOpen}.
De-conditioning on $r_1$ and $r_2$, we have 
\begin{eqnarray}
\label{eq:DenomFFRopen}
&~&\int_{r_2=0}^{\infty}\int_{r_1=0}^{\infty}\left(2\pi\lambda_1r_1 e^{-\pi\lambda_1{r_1}^2}\right)\left(2\pi\lambda_2r_2e^{-\pi\lambda_2{r_2}^2}\right) g_d(r_1,r_2) \dd r_1 \dd r_2.
\end{eqnarray}

Now we turn our attention to the numerator which equals,
\[ \E \left[e^{\left(-\mu \hat{I_1} T{r_1}^{\alpha} \right)}\right] - \E \left[e^{\left(-\hat{I_1} \mu T {r_1}^{\alpha}\right)}\left(\epsilon_1e^{\left(-\bar{I}\mu T_1 \frac{{r_1}^{\alpha}}{P_1} \right)}+\epsilon_2e^{\left(-\bar{I} \mu T_2 \frac{{r_2}^{\alpha}}{P_2} \right)}\right)\right].\] 

Concentrating on the second term we observe that the expectation with respect to $\hat{I_1}$,$I_1$, and $I_2$ is the joint Laplace transform of $\hat{I_1}$, $I_1$, and $I_2$ evaluated at $(\mu T {r_1}^{\alpha},\mu T_1 \frac{{r_1}^{\alpha}}{P_1},\mu T_2 \frac{{r_2}^{\alpha}}{P_2})$. The joint Laplace transform $g_n(r_1,r_2):=\mathscr{L}_{\textrm{num}}\left(\mu T {r_1}^\alpha, \mu  T_1 \frac{{r_1}^\alpha}{P_1}, \mu T_2 \frac{{r_2}^\alpha}{P_2}\right)$ is
\[\E_{\hat{I_1},I_1,I_2} \left[\exp \left(-s_1\hat{I_1}\right)\left(\epsilon_1\exp\left(-s_2(P_1I_1+P_2I_2)\right)+\epsilon_2\exp\left(-s_3(P_1I_1+P_2I_2)\right)\right)\right].\]
Expanding the terms and applying a similar approach as before we have
\begin{align*}
g_n(r_1,r_2)
= \epsilon_1e^{-2\pi \left( \lambda_1 \xi_{1,1}\left(T,T_1,\alpha,\Delta\right)+ \lambda_2 \rho_{1,2}\left(T_1,\alpha\right)\right)} + \epsilon_2e^{-2\pi \left(\lambda_1 \xi_{2,1}\left(T,T_2/\gamma,\alpha,\Delta\right) + \lambda_2 \rho_{2,2}\left(T_2,\alpha\right)\right)},
\end{align*}
where $\xi_{a,b}\left(T,z,\alpha,\Delta\right)$ is given by \eqref{eq:xiClosed}.
De-conditioning on $r_1$ and $r_2$,  
\begin{equation}
\label{eq:OpenFFRnum}
\int_{r_2=0}^{\infty}\int_{r_1=0}^{\infty}\left(2\pi\lambda_1r_1 e^{-\pi\lambda_1{r_1}^2}\right)\left(2\pi\lambda_2r_2e^{-\pi\lambda_2{r_2}^2}\right) g_n(r_1,r_2) \dd r_1 \dd r_2.
\end{equation}

Finally, plugging \eqref{eq:ClosedFFRnum} and \eqref{eq:DenomFFRopen} into \eqref{eq:OpenStrictFFRBayes}, and substituting \eqref{eq:firstClosedFFR} for the first term of the numerator by definition and ${r_1}^2 = v$ we have \eqref{eq:OpenBayesStrictFFR}.
 
\section{\label{sec:AppD} Proof of SFR, Open Access Theorem}
A user $y$ with $\sir_1 < T_1$ and $\sir_2 < T_2$ is given a FFR sub-band with uniform probability $\frac{1}{\Delta}$, and experiences new fading power $\hat{{g_1}}$, transmit power $\beta P_1$, and out-of-cell interference $\bar{I} = \eta P_1I_1 + P_2I_2$. The CCDF of the edge user $\bar{\mathrm{F}}_{\mathrm{SFR,o}}(T)$ is now given by
\begin{equation}
\mathbb{P}\left(\frac{\beta P_1\hat{{g_1}} {r_1}^{-\alpha}}{\eta P_1\hat{I_1}+P_2\hat{I_2}} > T ~\bigg|~ \frac{P_1{g_1}{r_1}^{-\alpha}}{\bar{I}+P_2{g_2}{r_2}^{-\alpha}} < T_1,\frac{P_2{g_2}{r_2}^{-\alpha}}{\bar{I} + P_1{g_1}{r_1}^{-\alpha}} < T_2 \right)
\label{eq:OpenSFR}
\end{equation}
Using the method of Theorem \ref{thm:OpenFFRedge}, applying Bayes' rule we have the joint Laplace transform of $I_1$ and $I_2$ given $r_1$ and $r_2$,
\begin{align*}
&f_d(r_1,r_2)=
1-\epsilon_1e^{\left(-2\pi\lambda_1\left(\rho_{1,1}(\eta T_1,\alpha)+ \kappa\rho_{1,2}(\gamma T_1,\alpha)\right)\right)}-\epsilon_2e^{\left(-2\pi\lambda_1\left(\rho_{2,1}\left(\frac{\eta}{\gamma}T_2,\alpha\right)+\kappa\rho_{2,2}\left(T_2,\alpha\right)\right)\right)},
\end{align*}
\[\textrm{where}~\gamma = \frac{P_2}{P_1} ~,~ \epsilon_1 = \left(\frac{1}{T_1\gamma\frac{{r_1}^{\alpha}}{{r_2}^{\alpha}}+1}\right)~,~\epsilon_2 = \left(\frac{1}{T_2\left(\gamma \frac{{r_1}^{\alpha}}{{r_2}^{\alpha}}\right)^{-1}+1}\right),\textrm{and}~\rho_{a,b}\left(z,\alpha\right)~\textrm{given by \eqref{eq:rhoOpen}}.\]

De-conditioning on $r_1$ and $r_2$, we have 
\begin{eqnarray}
\label{eq:OpenSFRdenom}
&~&\int_{r_2=0}^{\infty}\int_{r_1=0}^{\infty}\left(2\pi\lambda_1r_1 e^{-\pi\lambda_1{r_1}^2}\right)\left(2\pi\lambda_2r_2e^{-\pi\lambda_2{r_2}^2}\right) f_d(r_1,r_2) \dd r_1 \dd r_2.
\end{eqnarray}
Again, following the method of Theorem \ref{thm:OpenFFRedge}, we observe that the numerator of \eqref{eq:OpenSFR} is given by 
\begin{eqnarray}
\label{eq:numOpenSFR}
&~&\pi \lambda_1 \int_0^\infty e^{-\pi \lambda_1 v\left(1+\rho(\frac{\eta}{\beta} T,\alpha)+2\kappa\psi(\frac{\gamma}{\beta}T,\alpha)\right)} \dd v \nonumber\\
&-& \int_{r_2=0}^{\infty}\int_{r_1=0}^{\infty}\left(2\pi\lambda_1r_1 e^{-\pi\lambda_1{r_1}^2}\right)\left(2\pi\lambda_2r_2e^{-\pi\lambda_2{r_2}^2}\right) f_n(r_1,r_2) \dd r_1 \dd r_2,
\end{eqnarray}
\begin{align*}
\textrm{where}~ f_n(r_1,r_2)= &\epsilon_1e^{-2\pi\lambda_1 \left(\zeta_{1,1}\left(T,T_1,\alpha,\Delta,\beta,\eta\right)+\kappa\psi\left(\frac{\gamma}{\beta}T,\alpha\right)+ \kappa \rho_{1,2}\left(\gamma T_1,\alpha\right)\right)}\\
&+ \epsilon_2e^{-2\pi \lambda_1 \left(\zeta_{2,1}\left(T,T_2/\gamma,\alpha,\Delta,\beta,\eta\right)+\kappa \psi\left(\frac{\gamma}{\beta}T,\alpha\right)+\kappa \rho_{2,2}\left(T_2,\alpha\right)\right)}.
\end{align*}
Thus plugging \eqref{eq:numOpenSFR} and \eqref{eq:OpenSFRdenom} back into \eqref{eq:OpenSFR} and substituting ${r_1}^2 = v$ we have \eqref{eq:SFRopenFinal}.
\linespread{1.15}
\bibliographystyle{IEEEtran}
\bibliography{FemtoFFRbib}

\newpage

\begin{figure}
	\centering
	\includegraphics[width=4.5in]{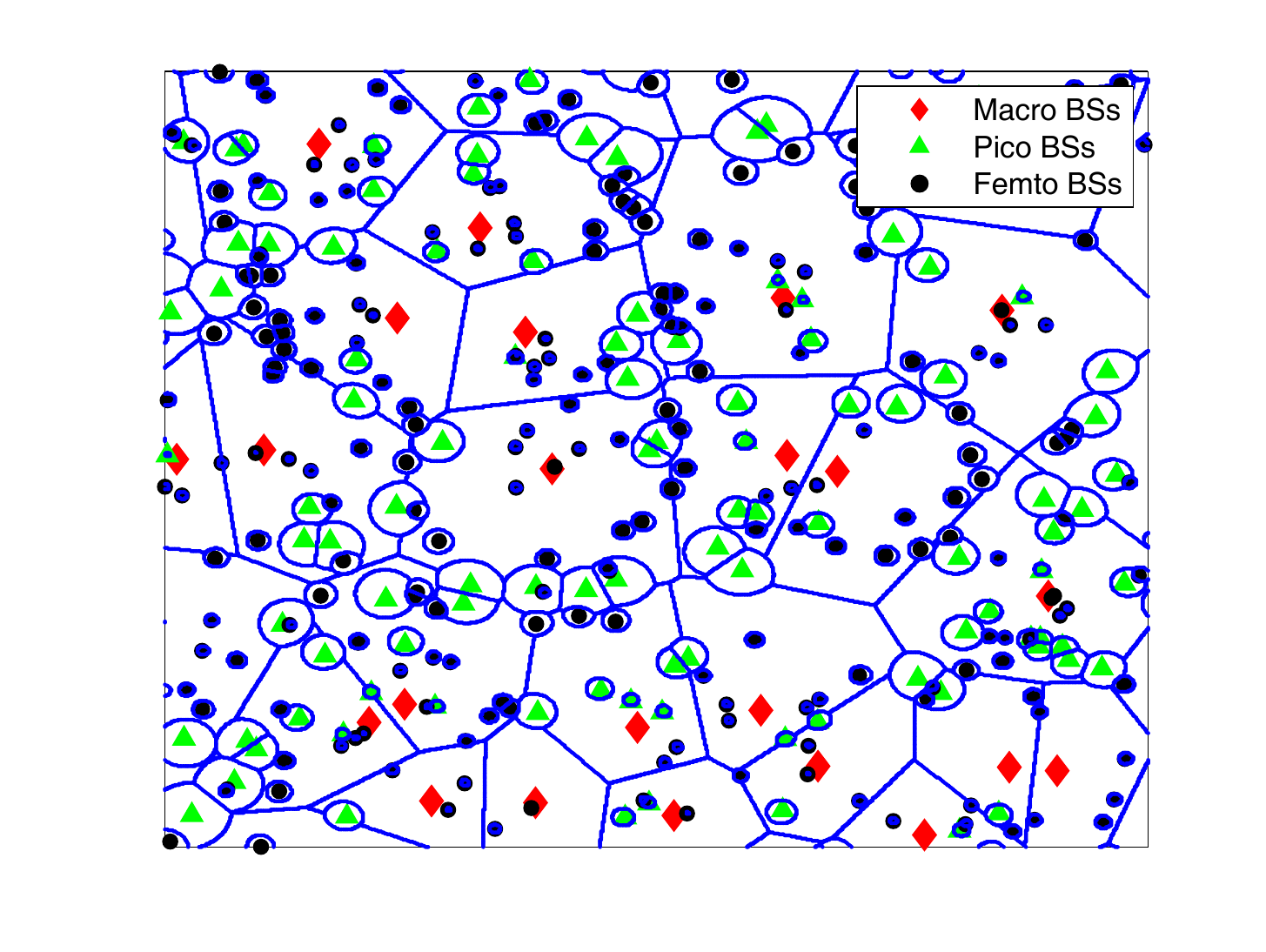}
	\caption{A realization of a Poisson distributed three-tier cellular network with coverage regions defined by the highest received power.}
   \label{fig:multiTier}
\end{figure}

\begin{figure}
	\centering
	\includegraphics[width=4in]{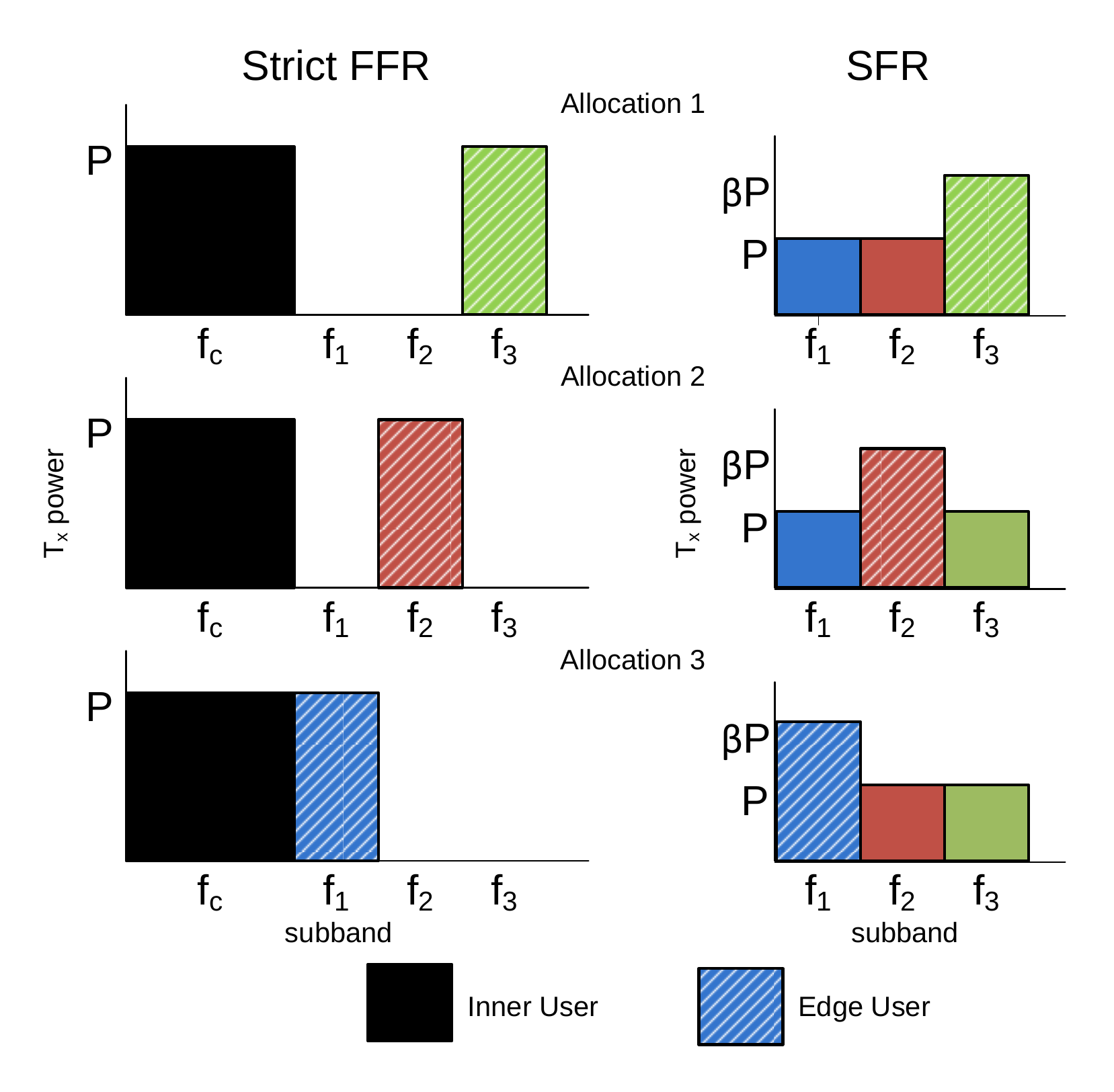}
	\caption{Strict FFR (left) and SFR (right) subband and transmit power allocations with $\Delta=3$ cell-edge reuse factor.}
	\label{fig:FFRmodels}
\end{figure}

\begin{figure}
	\centering
	\includegraphics[width=5in]{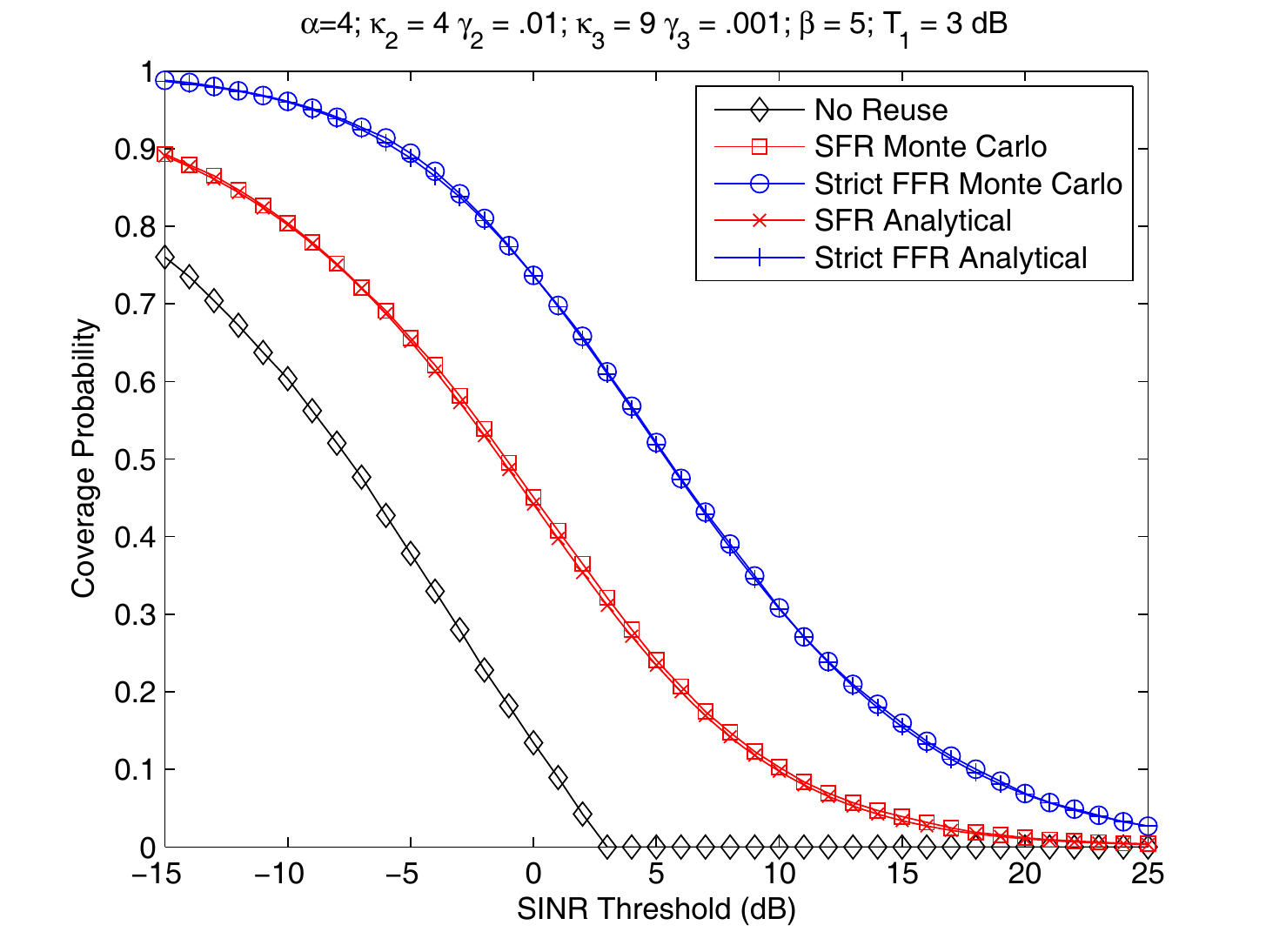}
	\caption{Downlink edge user $\sinr$ distributions for closed access with three tiers of APs.}
	\label{fig:ClosedEdge}
\end{figure}

\begin{figure}
	\centering
	\includegraphics[width=5in]{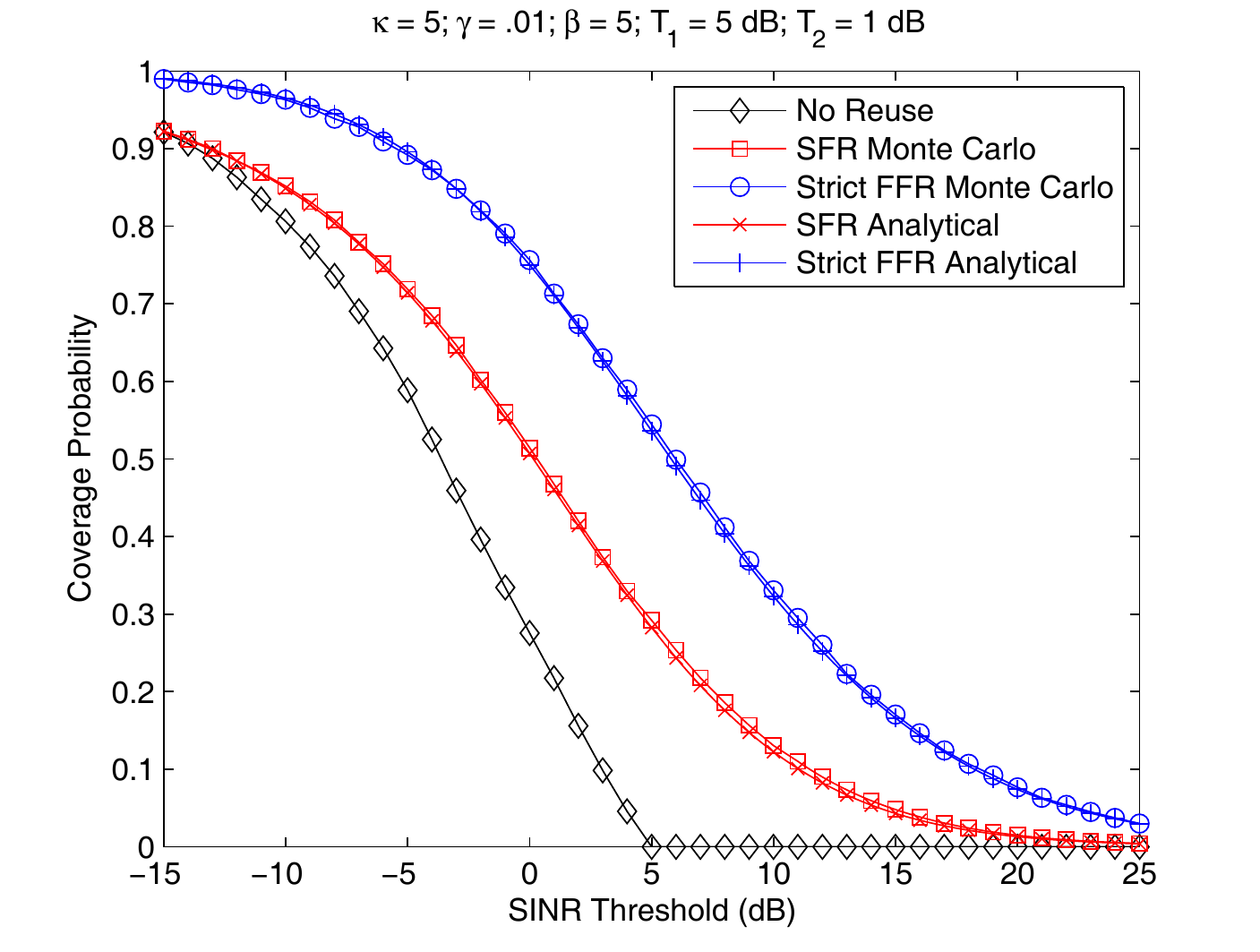}
	\caption{Downlink edge user $\sinr$ distributions for open access with two tiers of APs.}
	\label{fig:OpenEdge}
\end{figure}

\begin{figure}
	\centering
	\includegraphics[width=5in]{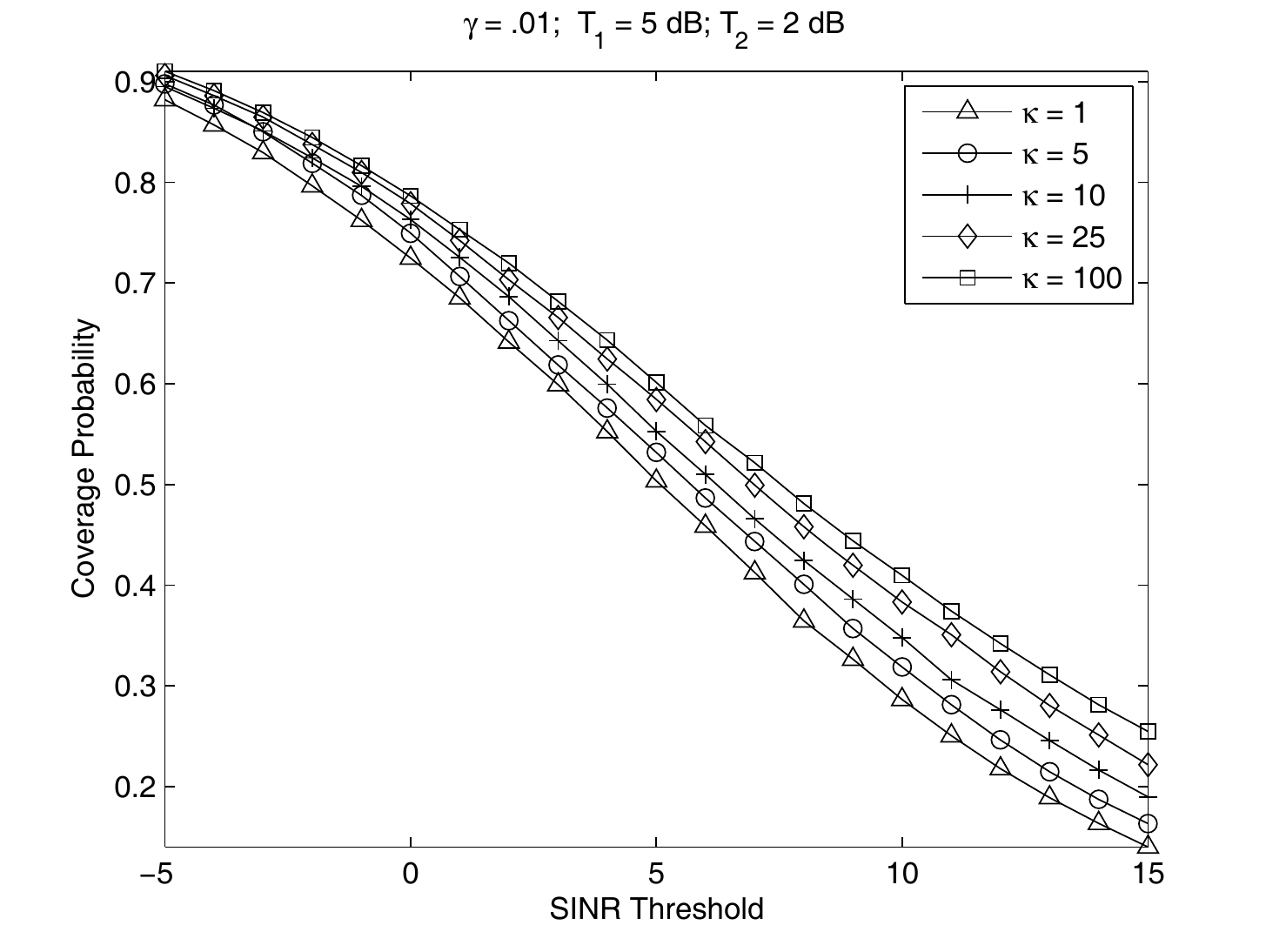}
	\caption{Downlink edge user $\sinr$ distributions for Strict FFR and closed access as a function of the tier density ratio $\kappa$.}
	\label{fig:MTkappa}
\end{figure}

\begin{figure}
	\centering
	\includegraphics[width=5in]{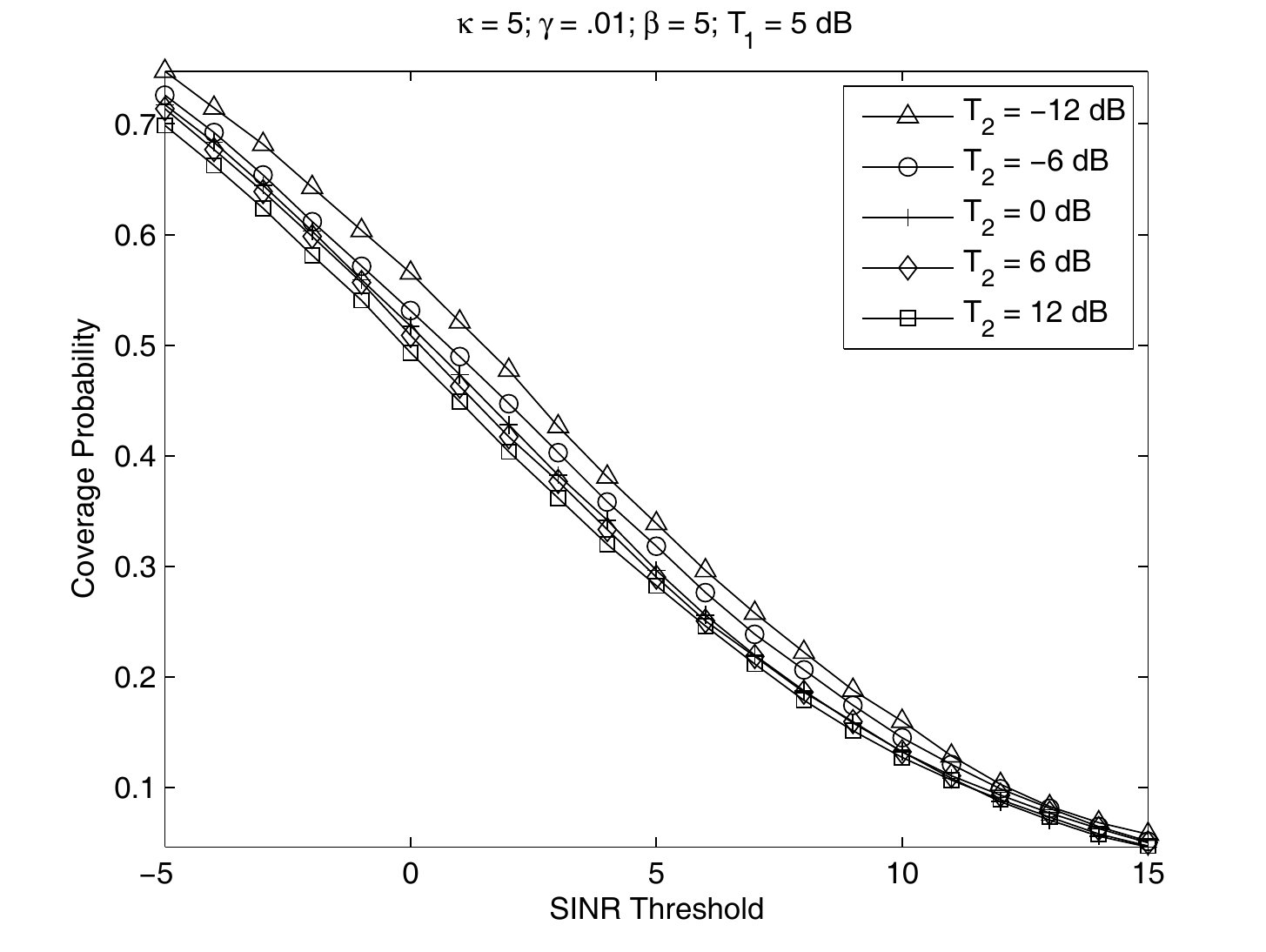}
	\caption{Downlink edge user $\sir$ distributions for SFR and open access as a function of $T_2$.}
	\label{fig:MTt2}
\end{figure}

\end{document}